\documentclass[english,aps,preprint]{revtex4}
\usepackage[utf8]{inputenc} 
\usepackage{amsmath} 
\usepackage{amsfonts}
\usepackage{amssymb}
\usepackage[english]{babel}
\usepackage{xcolor}
\definecolor{darkred}{rgb}{0.4,0.15,0.15}
\definecolor{bluegreen}{rgb}{0.20, 0.38, 0.5}
\usepackage{hyperref}
\hypersetup{
  colorlinks   = true,
  urlcolor     = blue,
  linkcolor    = {bluegreen},
  citecolor   = red
}
\makeatletter
\renewcommand{\thefigure}{\ifnum \c@section>\z@ \thesection.\fi
 \@arabic\c@figure}
\@addtoreset{figure}{section}
\makeatother
\numberwithin{equation}{section}
\usepackage{graphicx}
\usepackage{float}
\restylefloat{table} 
\newcommand{\pic}[3] {
\begin{figure}[H]
\centering
\includegraphics[scale=#1]{#2}
\caption{#3}
\end{figure}
}

\newcommand{\diff}{\operatorname{d}\!}

\newcommand{\amin}[2]{\underset{#1}{\operatorname{argmin}} \; \left\{ #2 \right\} } 

\newcommand{\prox}[2]{\operatorname{prox}_{#1}{\left(#2\right)}}

\newcommand{\sign}[1]{\operatorname{sign}\left( #1 \right)}

\newcommand{\id}{\operatorname{Id}}
\newcommand{\norm}[1]{\left\| #1 \right\|}

\usepackage[T1]{fontenc}
\usepackage{amsthm}
\newtheorem{prop}{Property}
\usepackage{algorithm}
\usepackage{algpseudocode}
\usepackage{comment}
\renewcommand{\algorithmicrequire}{}

\usepackage{setspace}
\newcommand{\ew}{\odot}
\usepackage{subfigure}

\begin{document} 
\title{A conjugate subgradient algorithm with adaptive preconditioning for LASSO minimization}
\author{Alessandro Mirone, Pierre Paleo}
\affiliation{European Synchrotron Radiation Facility, BP 220, F-38043 Grenoble Cedex, France}
\date{\today}
\begin{abstract}
This paper describes a new efficient conjugate subgradient algorithm which minimizes a convex function containing a least squares fidelity term
and an absolute value regularization term. 
This method is successfully applied to the inversion of ill-conditioned linear problems, in particular for computed tomography with the dictionary learning method. 

A comparison with other state-of-art methods shows a significant reduction of the number of iterations, which makes this algorithm appealing for practical use.
\end{abstract}

\maketitle

\section{Introduction}
Almost every field of science has, at some point, to tackle the linear inverse problem characterized by a matrix $A$.
In this problem, the observations vector $b$ can be expressed as
\begin{equation}\label{pbinverse}
b = A \tilde{x} + \epsilon
\end{equation}
where $\tilde{x}$ is the unknown signal to recover, $A$ the process matrix, and $\epsilon$ is some unknown noise.

The common Bayesian approach is to model the noise $\epsilon$ as an zero-mean Gaussian process of variance $\sigma^2 \id$, and the unknown variable $x$ as another random process.
If the signal $\tilde{x}$ is theoretically given, then the quantity $A \tilde{x}$ is deterministic; thus $b = A \tilde{x} + \epsilon$ is a random process.
More precisely, since $\epsilon \sim \mathcal{N}(0, \sigma^2)$, then $b \sim \mathcal{N}\left( b - A \tilde{x}, \sigma^2\right)$.
The likelihood function of $b$ is then given by
\begin{equation}
p(b \, | \, x) = \dfrac{1}{\sigma\sqrt{2\pi}} \exp\left( -\frac{\norm{b - Ax}_2^2}{2\sigma^2} \right)
\end{equation}
where $\norm{w}_2^2 = w^T w$ is the squared Frobenius norm of $w$, that is, the sum of the squared components.
Now since $x$ is unknown, the Bayesian approach consists in modeling it as another ramdom process.
If the signal is sparse in the representation $D  x$, where the columns of $D$ are the vectors of the basis, 
we can approximate the {\it a-priori} probability of $x$ using a Laplacian distribution, which implements the
sparsity inducing $L_1$ norm \cite{bach_sparsity_inducing} : 
\begin{equation}
p(x) = \dfrac{\beta}{2} \exp\left( -\beta \norm{D  x}_1 \right)
\end{equation}
where $\norm{w}_1$ is the $L_1$ norm of $w$, that is, the sum of the components absolute values.
The posterior probability $p(x \, | \, b)$ conditional on the observation vector $b$ then reads
\begin{equation}
p(x \, | \, b) = \dfrac{p(b \, | \, x) p(x)}{p(b)} \, \sim \,  p(b \, | \, x) p(x) \,\sim\, \exp\left( - \frac{\norm{A x - b}_2^2}{2\sigma^2} - \beta \norm{D x}_1 \right)
\end{equation}
Eventually, the Maximum A Posteriori (MAP) approach amounts to minimizing the log-Likelihood
\begin{equation}
\mathcal{L}(x \, | \, b) = \dfrac{1}{2}\norm{Ax-b}_2^2 + \beta \norm{D x}_1 \label{analysis_framework}
\end{equation}
which is the least squares formulation of \eqref{pbinverse} with a L1 norm regularization.
Notice that this penalty comes from the assumption made on the distribution of the values of $x$. Assuming normally distributed values of $x$ would have led to the Tikhonov regularization \cite{Tikhonov} (L2 norm).
L2-L1 minimization naturally arises in numerous applications when it comes to determine a solution with sparsity constraints.
In signal processing, one can cite deconvolution, image zooming, image inpainting, motion estimation \cite{chambollepock} and even tomographic reconstruction \cite{convex_prototyping}.

Generally speaking, L2-L1 is a special instance of the minimization problem 
\begin{equation}\label{basicpb}
\amin{x}{F(x) = f(x) + g(x)}
\end{equation}
where $F$ is purposely split into a \textit{convex, smooth} part $f$, and a \textit{convex, possibly non-smooth} part $g$.
This formulation is widely used for proximal splitting methods \cite{proxsplitting}, which rely on the computation of the so-called proximal operator

\begin{equation}
\prox{g}{x} =  \left( \id + \partial g\right)^{-1} (x) = \amin{y}{\frac{1}{2}\norm{x-y}_2^2 + g(x)}
\end{equation}
where $\partial g$ is the subdifferential of $g$ : 
\begin{equation}\label{subderivative}
\partial g(x) = \left\{ d \;| \; \forall\, y, \;\;  d^T ( y- x) \leq g(y)-g(x) \right\}.
\end{equation}
The subdifferential is set-valued where $g$ is not differentiable, and single-valued otherwise. For example, we have $\partial \norm{\cdot}_1(x) = \sign{x}$ if $x \neq 0$, and $\partial \norm{\cdot}_1 (0) = [-1, 1]$.

The case of L2-L1 minimization is a special instance of \eqref{basicpb}, where
\begin{equation}
\begin{aligned}
f(x) &= \dfrac{1}{2}\norm{A x - b}_2^2 \\
g(x) &= \beta \norm{D x}_1
\end{aligned}
\end{equation}
An alternative formulation to \eqref{analysis_framework} is the \textit{synthesis formulation}
\begin{equation}\label{synthesis_framework}
\amin{w}{\norm{A H w - b}_2^2 + \beta \norm{w}_1}
\end{equation}
and is celebrated as the least absolute shrinkage and selection operator (LASSO) \cite{lasso},
while \eqref{analysis_framework} implements, at variance with \eqref{synthesis_framework}, an \textit{analysis approach}.

The formulation \eqref{analysis_framework} corresponds to a linear inverse problem where $D x$ is constrained to be sparse. 
An example is the Total Variation regularization \cite{ROF}: $D = \norm{\nabla x}_1$.
In the formulation \eqref{synthesis_framework}, the solution $x = H w$ is synthesized from the coefficients $w$; these coefficients are constrained to be sparse in some domain.
An example is the Wavelet denoising for $A = \id$. 
These two approaches are equivalent if $D$ is an orthonormal transform (and then $H = D^*$ the hermitian conjugate of $D$) \cite{analysis_synthesis}.
However, in most cases, the theory and algorithms are more difficult in the analysis formulation.
In proximal splitting methods, the computation of $\operatorname{prox}_g$ is straightforward in the formulation \eqref{synthesis_framework} ($g = \norm{\cdot}_1$), but not trivial in the formulation \eqref{analysis_framework} ($g = \norm{D \cdot }_1$).

An alternative to proximal splitting methods is to adapt the functional $F$ in \eqref{basicpb}
in order to use fast optimization algorithms like Newton or conjugate gradient. 
It usually boils down to smoothing the regularization term $g(x)$.
However, such approaches converge to an approximate solution of \eqref{basicpb}, which can be an issue if high sparsity constraint should be met. 

We present in this work an algorithm, based on a new conjugate sub-gradient method optimized for LASSO minimization.
In the next section, after a brief recall of the conjugate gradient algorithm, we derive our algorithm.
Section \ref{applications} illustrates the applications with numerical examples :  one for a very ill-conditioned matrix,
and another for tomographic reconstruction with the dictionary-learning regularization. 
The convergence of this conjugate subgradient algorithm is compared to to the more general Nesterov \cite{Nesterov:1983wy} method.



\section{A conjugate subgradient algorithm}
\subsection{The nonlinear conjugate gradient algorithm}
In this section, we settle the notations by recalling the standard conjugate gradient algorithm.

Let $x$ denote the (vector) variable of the function $F$. 
For the remainder of this paper, the functional to minimize is $F(x) = f(x) + g(x)$ with $f(x) = \frac{1}{2} \norm{A x - b}_2^2$ and $g(x) = \beta \norm{x}_1$, so the optimization problem is
\begin{equation}\label{L1L2}
\amin{x}{F(x) = \frac{1}{2}\norm{A x - b}_2^2 + \beta \norm{x}_1}
\end{equation}

The conjugate gradient algorithm builds a set of conjugate directions $\left(p_k \right)_{k = 1\ldots n}$ where $n$ is the number of iterations.
Once the conjugate direction $p_k$ at iteration $k$, the variable is updated with $x_{k+1} = x_k + \alpha_k p_k$. 
The scalar $\alpha_k$ is the step size at iteration $k$, computed with a line search. 
The gradient of $F$ is then evaluated in $x_{k+1}$ to compute the next conjugate direction $p_{k+1}$.
The computation of $p_{k+1}$ actually only depends on the previous direction, which makes the conjugate gradient algorithm practically usable.

For a differentiable function $F$, the standard conjugate gradient is given by Algorithm \ref{alg:cg}.

\begin{algorithm}[H]\setstretch{1.35}
\caption{Conjugate gradient}\label{alg:cg}
\algorithmicrequire
$F$ : differentiable function\\
$n$ : number of iterations
\begin{algorithmic}[1]
\Procedure {conjGrad}{$F$, $n$}
\State Compute an initial guess $x_0$ 
\State $g_0 = -\nabla F(x_0)$ \Comment{Steepest direction at iteration $0$}
\State $p_0 = g_0$
\For {$k \leftarrow 0, n$}
	\State $\alpha_k = \amin{\alpha}{F(x_k + \alpha p_k)}$ \Comment{Line search}
	\State $x_{k+1} = x_k + \alpha_k p_k$ \Comment{Update variable} 
	\State $g_{k+1} = -\nabla F(x_{k+1})$ \Comment{Update Steepest direction}
	\State $\beta_k = \dfrac{g_{k+1}^T (g_{k+1} - g_k)}{g_k^T g_k}$ \Comment{Update $\beta$, for example with the Polak-Ribiere rule}
	\State $p_{k+1} = g_{k+1} + \beta_k p_k$ \Comment{New conjugate direction}	
\EndFor
\State \Return $x_n$
\EndProcedure
\Statex
\end{algorithmic}
\end{algorithm}

\subsection{From conjugate gradient to conjugate subgradient}

In the basic subgradient method
\begin{equation}\label{sg}
x_{k+1} = x_k - \gamma_k p_k \qquad p_k \in \partial F(x_k)
\end{equation}
the direction $p_k$ is any subgradient $\partial F(x_k)$, 
which is a drawback of this method since there is no indication of which subgradient should be chosen.
As a result, the conjugate subgradient is not a descent method: the objective function can increase during the optimization process \nolinebreak\cite{subgradmethods}.

To build an algorithm based on the conjugate gradient, one has to define an unique descent direction at each iteration, which means choosing between all the possible subgradients $\partial F$ when $F$ is not differentiable.

The basic idea is to rely on the quadratic part $\nabla f$ of the gradient. Once the gradient of the smooth part $\nabla f (x)$ is calculated, the subgradient of the L1 part $g$ is evaluated with :
\begin{equation}\label{completegrad}
\partial g (x) =
\begin{cases}
\sign{x} & \text{if } x \neq 0 \\
\sign{\nabla f (x)} & \text{if } x = 0
\end{cases}
\end{equation}
Notice that using \eqref{completegrad}, the subderivative of $F = f + g$ is always single-valued.
The motivation of such a choice is that when the variable $x$ comes near the singularity of $g = \norm{\cdot}_1$, 
every direction (subgradient) is possible.
The idea is then to go in the same direction than the quadratic term is ``pushing" to.


The use of \eqref{completegrad} to compute the subgradient enables to solve the indecision of which subgradient should be chosen, and makes possible the construction of a conjugate directions basis. 
The standard Polak-Ribiere method can be used to update the conjugate direction from the previous directions.

A crucial point for the convergence rate is the use of a preconditioner.
In our method, the preconditioner relies on the magnitude of the quadratic part of the gradient $\nabla f$. 

From the variables $x_{k+1}, \, p_k, \, q_k$ (see Algorithm \ref{alg:cg}), three new preconditioned variables $\overline{x}_{k+1}, \, \overline{p}_{k+1}, \, \overline{q}_{k+1}$ are built with the following preconditioner :
\begin{equation}\label{update_prec}
\left\lbrace
\begin{aligned}
D &= \begin{cases}
1 & \text{if } \left| \nabla f (M_k  \ew x_{k+1}) \right| < \beta  \text{ and } {x_k} \cdot x_{k+1} < 0 \\
0 & \text{otherwise}
\end{cases} \\
M_{k+1} &= \min\left( M_k \cdot \left( 1 - \gamma D + \delta (1 - D) \right) , \, 1 \right) \\
S_{k+1} &= \begin{cases}
0 & \text{if } \left| \nabla f (M_k  \ew x_{k+1}) \right| < \beta  \text{ and } |x| < \varepsilon \\
1 & \text{otherwise}
\end{cases}\\
V_{k+1} &= \dfrac{M_{k+1}}{M_k}
\end{aligned}
\right.
\end{equation}
\begin{equation}\label{update_x_bar}
\left\lbrace
\begin{aligned}
\overline{x}_{k+1} &= \frac{x_{k+1}}{V_{k+1}} \cdot S_{k+1} \\
\overline{p}_{k+1} &= p_{k} \cdot V^\alpha_{k+1} \cdot S_{k+1} \\
\overline{q}_{k+1} &= q_{k} \cdot V_{k+1} \cdot S_{k+1}
\end{aligned}
\right.
\end{equation}
all the operation being componentwise except for the argument 
of $f$ which is obtained with the componentwise multiplication $\ew$ between the vector $x_{k+1}$
and the vector of preconditioning multiplying factors.

The rationale of this preconditioner can be summarized as follow : 
\begin{itemize}
\item When the gradient magnitude of the quadratic part $\nabla f$ is important, the components of the variables are updated as in the conjugate gradient method -- without variable substitution -- since the quadratic part is predominant over the non-smooth part.
\item When $\left| \nabla f \right|$ is small, the standard conjugate gradient method
 would be disturbed by frequent crossings of regions where the gradient of $g$ is discontinuous. 
The rule used is that the preconditioning factors  are increasingly shrunk by a factor $\gamma < 1$ as long as they should be updated. 
The criterion is to check if the previous preconditioned variable ($\overline{x}_k$) and the variable updated after the line search ($x_{k+1}$) have an opposite sign. This variable substitution is implemented by the coefficient vector $M_k$.
\item The exponent $a$, used in the determination of the vector $\overline{p}_{k+1}$
is a tunable number. The vector $\overline{p}_{k+1}$ is used in the composition of the $p_{k+1}$ descent direction(see Algorithm \ref{alg:cg}). By using a  number $a>-1$ we tend to avoid constructing descent directions which bring us too fast to non-smooth regions. Keeping $a=-1$ corresponds to 
using the previous descend direction as in standard conjugate gradient method.

\item Another rule is that during this phase (small quadratic gradient), the components which are ``small enough" (below a threshold $\epsilon$) are set -- and will remain as long as the force on them is weak-- to zero.
This rule is especially important for the convergence toward solutions with high sparsity.
This rule is implemented by the matrix $S_k$.
\end{itemize}


The conjugate subgradient algorithm for LASSO optimization is given by Algorithm \ref{alg:csg}

\begin{algorithm}[H]\setstretch{1.35}
\caption{Conjugate subgradient}\label{alg:csg}
\algorithmicrequire
$F$ : function to optimize, $F(x) = f(x) + g(x)$ with $f$ the quadratic part and $g$ the L1 part \\
$\gamma, \delta, \epsilon$ : parameters for update the preconditioner (see \eqref{update_prec}) \\
$n$ : number of iterations\\
\begin{algorithmic}[1]
\Procedure {conjSubGrad}{$F$, ($\gamma$, $\delta$, $\epsilon$), $n$}
\State Compute an initial guess $\overline{x}_0$ 
\State $g_0 = -\nabla F(x_0)$ \Comment{Steepest direction at iteration $0$}
\State $p_0 = g_0$
\State $M_0 = 1$ \Comment{Element-wise}
\For {$k \leftarrow 0, n$}
	\State $q_k = M_{k} \ew A^T A (M_{k} \ew p_k) $
	\State Compute $\alpha_k = \amin{\alpha}{F(M_{k} \ew (\overline{x}_k + \alpha p_k))}$
	\State $x_{k+1} = \overline{x}_k + \alpha_k p_k$
	\State Update preconditioners $(M_{k+1}, \, S_{k+1}, \, V_{k+1})$ using \eqref{update_prec}
	\State Update $(\overline{x}_{k+1}, \, \overline{p}_{k+1}, \, \overline{q}_{k+1})$ using \eqref{update_x_bar}
	\State $g_{k+1} = -\nabla F (\overline{x}_{k+1} \ew M_{k+1}) \ew S_{k+1}\ew M_{k+1}$
	\State $\beta = -\dfrac{\overline{q}_{k+1}^T g_{k+1}}{\overline{q}_{k+1}^T \overline{p}_{k+1}}$
	\State $p_{k+1} = g_{k+1} + \beta \overline{p}_{k+1}$

\EndFor
\State \Return $x_n$
\EndProcedure
\Statex
\end{algorithmic}
\end{algorithm}

\subsection{Line search}
The line search is a crucial step of gradient methods. The variables are updated with the previously computed conjugate direction $p_k$.
The step $\alpha_k$ in this direction should be such as
\begin{equation}\label{alpha1}
\alpha_k = \amin{\alpha}{F(M_{k+1}\ew x_{k+1})} \qquad \text{with } \, x_{k+1} = x_k + \alpha p_k
\end{equation}
The computation of \eqref{alpha1} can be done ``blindly" with a generic line search, but here one can benefit from both the quadratic nature of $f$ and the convex property of $g$. We discuss how to do it in this session, discarding for conciseness, and without loss 
of generality,  the notation of preconditioner vector $M$.

Regarding the quadratic part $f$, it is easily shown that
\begin{equation}
f(x_k + \alpha p_k) = \frac{1}{2}\norm{A (x_k + \alpha p_k) - b}_2^2  = a_2 \alpha^2 + a_1 \alpha + a_0
\end{equation}
\begin{equation*}
\text{with } a_2 = \frac{1}{2} p_k^T A^T A p_k , \;\,
a_1 = p_k^T A^T \left( A x_k - b \right), \;\,
a_0 = \frac{1}{2}\left( x_k^T A^T A x_k + b^T b \right)
\end{equation*}
The coefficients $a_2$ and $a_1$ can be computed once for all before the line search ; actually, they are also used elsewhere in the algorithm so they have to be computed anyway.
The evaluation of $\frac{\diff f}{\diff \alpha}$, the derivative of $f$ with respect to the scalar $\alpha$, only requires these two coefficients, and thus has virtually no cost.

Another interesting property of smooth quadratic function $f(x) = \norm{A x - b}_2^2$ is
\begin{equation}\label{update_grad_smooth}
\nabla f(x_{k+1}) = \nabla f (x_k) + \alpha_k A^T A p_k
\end{equation}
The quantity $A^T A p_k$ is also reused, for example with the computation of $p_k^T A^T A p_k$. Hence the update of the gradient $\nabla f(x_{k+1})$ from the previous gradient $\nabla f(x_k)$ is cheap.

For a smooth quadratic function, the line search is straightforward:
\begin{align*}
0 = \dfrac{\diff f}{\diff \alpha} &= \nabla f (x_{k+1})^T \cdot \dfrac{\diff}{\diff \alpha} x_{k+1} \\
	&= p_k^T \left( \nabla f (x_k) + \alpha A^T A p_k \right) \qquad \text{using \eqref{update_grad_smooth}}\\	
\end{align*}
which gives
\begin{equation}\label{alpha_smooth}
\alpha_k = \dfrac{-p_k^T \nabla f(x_k)}{p_k^T A^T A p_k}
\end{equation}

Now, getting back to the whole function $F = f + g$, a one-step line search like \eqref{alpha_smooth} is not possible since one cannot extract $\alpha$ from $\partial g (x_{k+1})$.
However, due to the convexity of $g$, an upper bound of $\alpha_k$ can be computed using the following property :
\begin{prop}\label{prop:pg}
For all $k$, we have $p_k^T \partial g(x_{k+1}) \geq p_k^T \partial g(x_k)$.
\end{prop}
\begin{proof}
Since $g$ is convex, every component $\partial g^i$ of its subgradient is increasing. Thus, we have $\partial g(x_{k+1})^i \geq \partial g(x_k)^i$ if and only if $x_{k+1}^i \geq x_k^i$, i.e $p_k^i \geq 0$ (since $\alpha_k \geq 0$). Thus :
\begin{itemize}
\item If $p_k^i \geq 0$, then $x_{k+1}^i = x_k^i + \alpha_k p_k^i \geq x_k^i$, so $\partial g(x_{k+1})^i \geq \partial g (x_k)^i$, 
so $p_k^i \cdot \partial g(x_{k+1}) \geq p_k^i \cdot \partial g (x_k)$.
\item Similarly, if $p_k^i \leq 0$, then $\partial g (x_{k+1})^i \leq \partial g(x_k)^i$ so $p_k^i \cdot \partial g (x_{k+1})^i \geq p_k^i \cdot \partial g(x_k)^i$.
\end{itemize}
Doing the scalar product, we have in any case $p_k^T \partial g(x_{k+1}) \geq p_k^T \partial g(x_k)$
\end{proof}
Using this property, we can derive the same calculation as for \eqref{alpha_smooth} :
\begin{equation}\label{alpha_up_calc}
\begin{aligned}
0 = \dfrac{\diff F}{\diff \alpha} &= \dfrac{\diff f}{\diff \alpha} + \dfrac{\diff g}{\diff \alpha} \\
	&= p_k^T \nabla f (x_k) + \alpha p_k^T A^T A p_k + p_k^T \partial g (x_{k+1}) \\
	& \geq p_k^T \left( \nabla f (x_k) + \partial g (x_k) \right) + \alpha p_k^T A^T A p_k
\end{aligned}
\end{equation}
Thus
\begin{equation}\label{alpha_up}
\alpha_k \leq \alpha_k^u = \dfrac{-p_k^T \partial F (x_k)}{p_k^T A^T A p_k}
\end{equation}
For the last inequality in \eqref{alpha_up_calc}, property \ref{prop:pg} has been applied.
The upper bound $\alpha_k^u$ is convenient for a line search using the bisection method.
For example, the line search can be done using the \textit{regula falsi} method at the beginning when the differentiable L2 part is predominant, and then the bisection method when the L1 part becomes more important.

\section{Applications}\label{applications}
In this section, numerical examples are provided to compare the convergence of this new method with Nesterov algorithm \cite{Nesterov:1983wy}, also known as FISTA \cite{beckteboulleIEEE}
which is a state-of-art convex non-smooth optimization method.

\subsection{Example on ill-conditioned matrix}
This example illustrates the convergence rate of the conjugate subgradient algorithm for problem \eqref{L1L2}, where the matrix $A$ is chosen to be ill-conditioned.
The code to compute this example can be found at \cite{matrice_folle} 
In this example, $A$ is a $1000 \times 1000$ symmetric matrix, with a condition number $\kappa = \frac{\sigma_{\text{max}}}{\sigma_{\text{min}}} = \frac{\lambda_{\text{max}}}{\lambda_{\text{min}}} = \frac{95.5}{1.61\cdot 10^{-14}} \simeq 5.93 \cdot 10^{15}$. 
The eigenvalues of $A$ are plotted on Figure \ref{Fig:eigens}.

\pic{0.5}{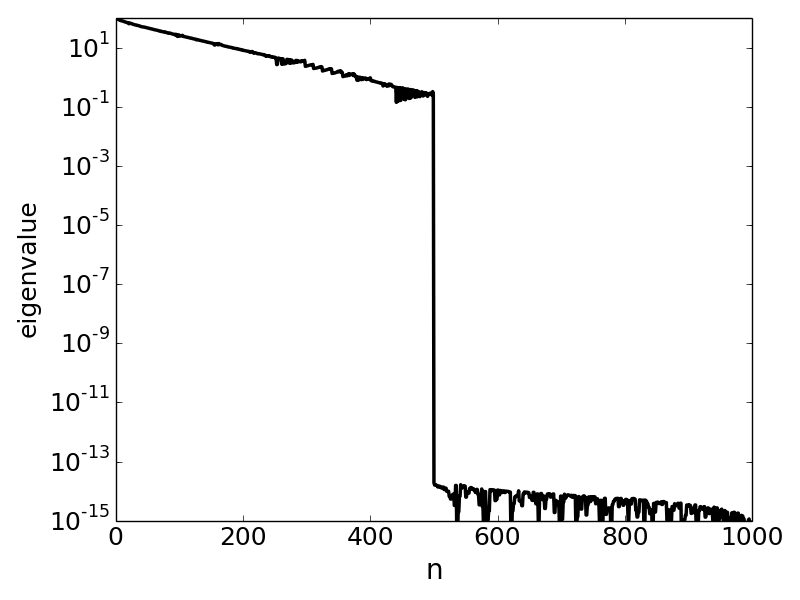}{Logarithmic plot of the eigenvalues of the matrix $A$\label{Fig:eigens}}
The algorithm was run with the parameters $\gamma = 0.85$ and $\delta = 0.04$, the regularization parameter was $\beta = 0.1$ and the exponent for direction $\overline p$ was $a=1$.
Figure \ref{Fig:fistavscsg} shows the objective function values $F(x)-F(x_\infty)$ for $2000$ iterations
for the two methods.
It can be seen that CSG achieves the solution in about 800 iterations, while FISTA needs much more iterations to converge.
Also, the objective function values are always smaller for CSG.

\pic{0.5}{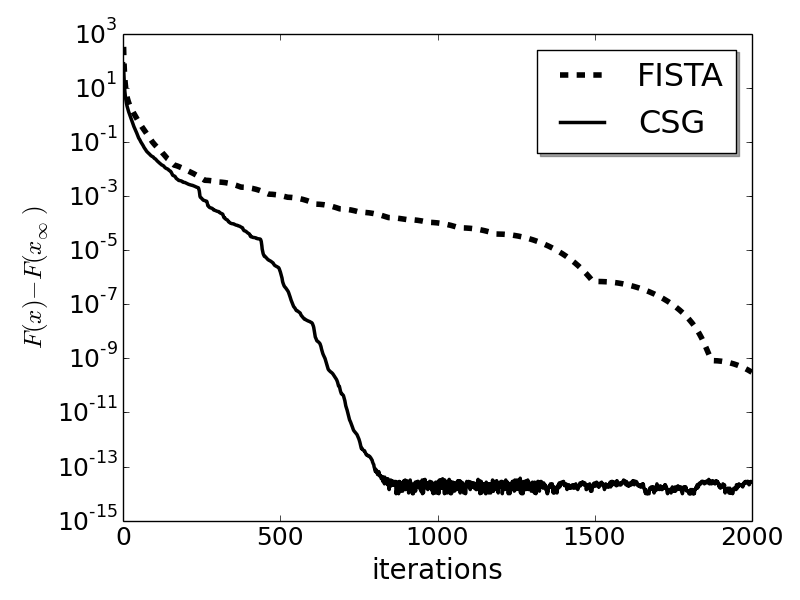}{Logarithmic plot of objective function values for CSG and Nesterov algorithm\label{Fig:fistavscsg}}

\subsection{Tomographic reconstruction with the dictionary-learning regularization and ring-artifacts correction}
Tomographic reconstruction is another example of linear inverse problem.
In the last years, an increasing interest was shown for iterative techniques with regularization, 
which can be seen as an extension of the standard Algebraic Reconstruction Technique and Simultaneous Iterative Reconstruction Technique.
These techniques bring many opportunities, for example
modeling more accurately the process, incorporating a priori knowledge on the volume and correcting artifacts.
A prominent application is the low-dose tomography reconstruction.

Iterative tomographic reconstruction amounts to an optimization problem
An example is the the \textit{total variation} reconstruction 
\begin{equation}\label{tomotv}
\amin{x}{\norm{P x - d}_2^2 + \beta \norm{\nabla x}_1}
\end{equation}
which penalizes the nonzero components of the gradient of the slice, promoting piecewise constant results.
Here $x$ denotes the slice (or volume) to be reconstructed, $P$ is the projection operator, $d$ is the acquired sinogram and $\beta$ is a factor weighting the sparsity of the gradient of the solution.
Another example is the \textit{dictionary learning} reconstruction
\begin{equation}\label{tomodl}
\amin{w}{\norm{P D w - d}_2^2 + \beta \norm{w}_1}
\end{equation}
which promotes the sparsity of the slice in an appropriate basis $D$ : either a learned dictionary \cite{DLpyhst2} or a Wavelet transform.

Notice that \eqref{tomotv} correspond to an analysis formulation while \eqref{tomodl} is a synthesis formulation, for which the conjugate subgradient can be applied.

In this example, the standard $512 \times 512$ test image \textit{Lena} was used.
According to the Nyquist criterion, $\frac{\pi}{2} 512 \simeq 800$ projections would be required to get an appropriate reconstruction quality with the Filtered Back Projection.
With iterative techniques promoting sparsity, this number can be dramatically decreased according to the Compressive Sensing theory \cite{candes_cs}.
Here only $80$ projections were used to demonstrate the abilities of the Dictionary Learning technique.
Additionally, rings artifacts were simulated by adding lines in the sinogram.
The lines values are not constant along the projection angle, which makes the problem more challenging.
To take the rings correction into account  \cite{ringsDLTV} , the reconstruction problem is written as \eqref{tomodl_rings}.

\begin{equation}\label{tomodl_rings}
\amin{w}{F(w) = \norm{P D w + 1 \times r^T - d}_2^2 + \beta \norm{w}_1 + \beta_r \norm{r}_1} 
\end{equation}

In this formalism, a ring vector $r$ is added to each projection line of the sinogram -- the rings artifacts are modeled as constant values along the projection angle in the sinogram.
The sinogram has dimensionality $(N_p, N)$ where $N_p$ is the number of projections and $N$ is the number of pixels in one dimension of the slice.
The operation $1 \times r^T$ consists in multiplying a $(N_p, 1)$ vector of ones with a $(1, N)$ vector $r$.

The functional \eqref{tomodl_rings}   was minimized with two techniques implemented in the PyHST2 code \cite{pyhst2} : Nesterov algorithm (FISTA) and this conjugate subgradient algorithm (CSG).
In this test, an over-complete dictionary has been used, resulting in an ill-conditioned problem which is a difficult
test case for optimization algorithms.
Moreover we observed that, for this kind of problem, the transfer of energy
from the reconstructed image to the auxiliary variables capturing the spurious artifacts ($r$) occurs in the 
final part of the convergence and is slow with the FISTA. The best convergency properties were obtained with $a=0$.

Figure \ref{Fig:vs_rings} shows the plot of the normalized objective function $F(w)-F(w_\infty)$ for $8000$ iterations.
Both methods converge to the same final value since the same functional $F(w)$ is minimized, but the last stage of the optimization process
is much faster for the conjugate subgradient algorithm.
Figure \ref{Fig:tomoresults} shows the reconstructed images with Filtered Back Projection and the Dictionary Learning technique, for parameters $\beta = 0.7$ and $\beta_r = 10$.
It can be noted that the rings artifacts are almost entirely removed, even with the simple ``constant rings" modeling.

{
\begin{figure}
\includegraphics[scale=0.4]{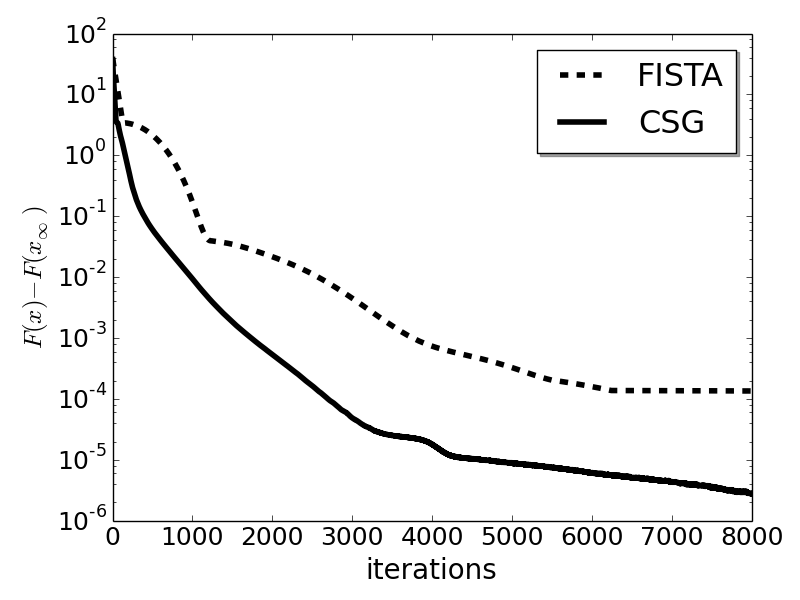}
\includegraphics[scale=0.4]{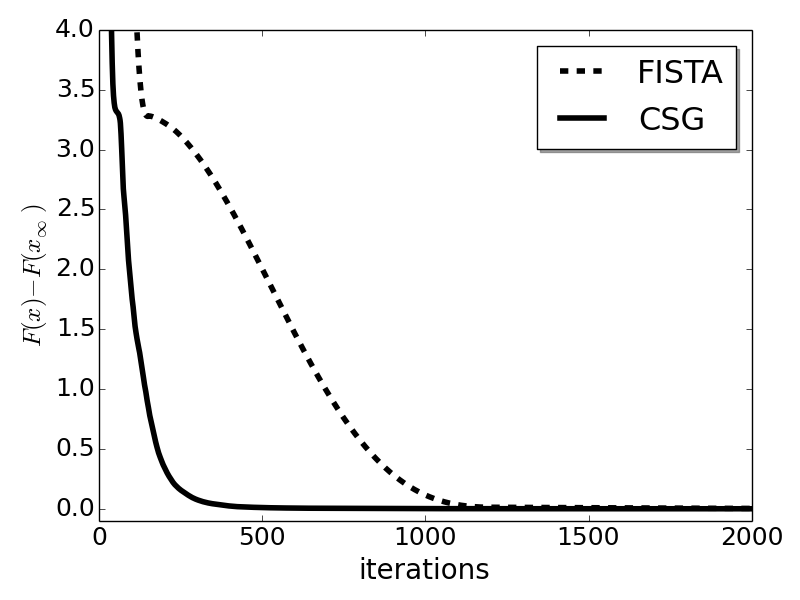}
\caption{Logarithmic and linear plots of the values of the objective function for both methods \label{Fig:vs_rings}}
\end{figure}
}

\begin{figure}[H]
     \begin{center}
        \subfigure[Phantom of Lena]{%
            \label{fig:first}
            \includegraphics[width=0.4\textwidth]{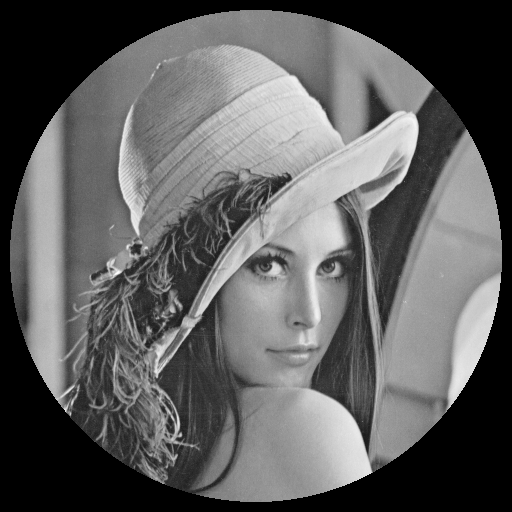}
        }%
        \subfigure[Filtered Back Projection]{%
           \label{fig:second}
           \includegraphics[width=0.4\textwidth]{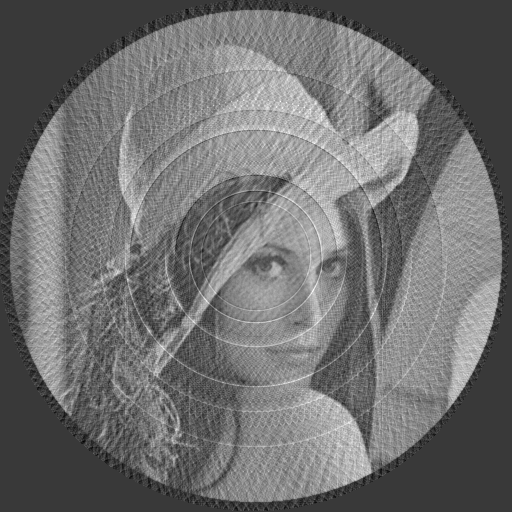}
        }
            \subfigure[Dictionary Learning]{%
           \label{fig:second}
           \includegraphics[width=0.4\textwidth]{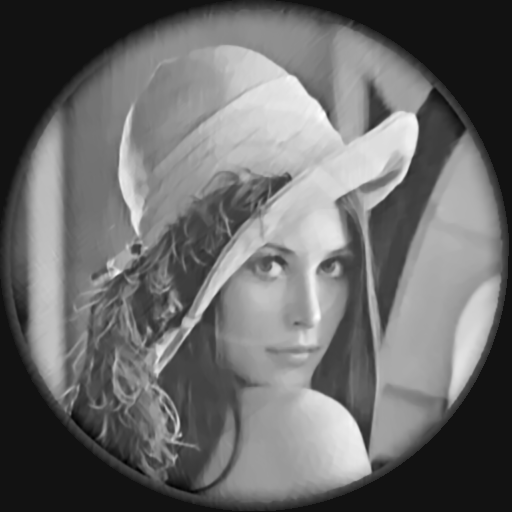}
        }
    \end{center}
    \caption{%
        Phantom of Lena reconstructed with $80$ projection angles.
        Lines were added to the sinogram to simulate ring artifacts.
     }%
   \label{Fig:tomoresults}
\end{figure}

\section{Conclusions}
We have presented a specialized Conjugate Sub Gradient method which we have tailored
for the LASSO minimization.
This method is fit to cope at the same time with the ill-conditioning of the LASSO matrix
and the discontinuities in the first derivative.
We have tested our method on two difficult cases and found excellent acceleration, outperforming state-of-the art algorithms.
An implementation of CSG can be found at \cite{matrice_folle}.

\section*{Acknowledgement}
We thank Jerome Lesaint which, during his stage from UJF, partecipated to the initial phase of the investigations,
studying the convergence properties of Conjugate Gradient on smoothed LASSO problems.

\section*{References}
\bibliographystyle{ieeetr}
\bibliography{biblio}

\end{document}